\documentclass[12pt, draftclsnofoot, onecolumn]{IEEEtran}

\usepackage{amsthm}
\usepackage{amsmath,amssymb,amsfonts}
\usepackage{psfrag}
\usepackage[numbers,sort&compress]{natbib}
\usepackage{flushend}
\usepackage{bbold}

\usepackage{color}

\newtheorem{corollary}{Corollary}
\newtheorem{theorem}{Theorem}

\newtheorem{lemma}{Lemma}

\newtheorem{remark}{Remark}

\ifCLASSINFOpdf
  \usepackage[pdftex]{graphicx}
  \usepackage{epstopdf}
  \DeclareGraphicsExtensions{.eps, .pdf, .jpeg}
\else
\fi

\ifCLASSINFOpdf
    \usepackage[pdftex]{graphicx}
\else
    \usepackage[dvips]{graphicx}
\fi

\def \O {{\mathcal O }}
\def \P {{\mathcal P }}

\begin{document}
\title{{ Physical Layer Security in Uplink NOMA Multi-Antenna Systems with Randomly Distributed Eavesdroppers}}

\author{\normalsize Gerardo Gomez, Francisco J. Martin-Vega, F. Javier Lopez-Martinez, Yuanwei Liu \\and Maged Elkashlan \vspace{-1.25cm}
\thanks{Gerardo Gomez, Francisco J. Martin-Vega and F. Javier Lopez-Martinez are with the Departamento
de Ingenier\'ia de Comunicaciones, Universidad de M\'alaga, M\'alaga 29071, Spain (e-mail: ggomez@ic.uma.es, fjmvega@ic.uma.es, fjlopezm@ic.uma.es).}
\thanks{Yuanwei Liu and Maged Elkashlan are with the Queen Mary University of London, London, UK (e-mail: yuanwei.liu@qmul.ac.uk, maged.elkashlan@qmul.ac.uk).}
\thanks{This work has been submitted to the IEEE for publication. Copyright may be transferred without notice, after which this version may no longer be accessible.}
}

\maketitle

\begin{abstract}
The physical layer security of uplink non-orthogonal multiple access (NOMA) is analyzed. A stochastic geometry approach is applied to analyze the coverage probability and the effective secrecy throughput (EST) of the $kth$ NOMA user, where a fixed or an adaptive transmission rate can be used. A protected zone around the legitimate terminals to establish an eavesdropper-exclusion area. We assume that the channel state information associated with eavesdroppers is not available at the base station. { We consider that the base station is equipped with multiple antennas.} The impact of imperfect successive interference cancellation is also taken into account in this work. Our analysis allows an easy selection of the wiretap code rates that maximizes the EST. Additionally, our framework also allows an optimum selection of other system parameters like the transmit power or the eavesdropper-exclusion radius.

\end{abstract}


\section{Introduction}

Non-orthogonal multiple access (NOMA) has recently been introduced as a new feature intended to increase the spectrum efficiency in the fifth generation (5G) networks \cite{Ding17,Liu16}. This technique allows serving multiple users simultaneously using the same spectrum resources at the cost of increased intra-cell interferences \cite{Tabassum16}. NOMA may use the power domain jointly with interference cancellation techniques to separate signals, exploiting the path loss differences among users. 

In uplink (UL) NOMA, a set of users transmits simultaneously their signals to their associated base station (BS). As a consequence, the received signal of a particular user suffers from intra-cluster interference, which is a function of the channel statistics of other users. In order to minimize such interference, the BS may apply successive interference cancellation (SIC) to decode signals. SIC technique requires that different message signals arrive to the receiver (BS) with a sufficient power difference so that SIC may be successfully applied. This is typically achieved in the downlink (DL) by means of different weights at the transmitter. However, since the UL channel gains already provide sufficient distinctness between the received signals, such weights are not necessary. In fact, the conventional UL transmit power control intended to equalize the received signal powers of users is not recommended for UL NOMA transmissions since it may remove channel distinctness \cite{Tabassum16}. 

SIC technique in UL NOMA works as follows. The BS first decodes the strongest signal by considering the signals from other users as noise. However, the user with the weakest signal enjoys zero intra-cluster interference since the BS has previously canceled interfering signals (considering ideal conditions). If we consider the possibility of a SIC failure, the error is propagated to all remaining messages. UL NOMA was firstly presented in \cite{Endo12}, by considering the minimum mean squared error (MMSE)-based SIC decoding at the BS. { An interesting survey on NOMA for 5G networks is presented in \cite{Ding17b}, which provides a comprehensive overview of the latest NOMA research results and innovations.} A novel dynamic power allocation scheme for DL and UL NOMA is proposed in \cite{Yang16}.  The outage performance and the achievable sum data rate for UL NOMA is theoretically analyzed in \cite{Zhang16b}. In \cite{Tabassum17}, a framework to analyze multi-cell UL NOMA with stochastic geometry is presented. In \cite{Rabee17}, the optimum received UL power levels using a SIC detector is determined analytically for any number of transmitters. 

The possibility of having a secure communication in NOMA-based scenarios is also a current hot topic. The presence of eavesdroppers is a classical problem in communication theory, ever since Wyner introduced the wiretap channel \cite{Wyner1975}. In the last years, the field of physical layer security over different scenarios has taken an important interest in the research community as a means to provide reliable secure communications, relaxing the complexity and complementing the performance of the required cryptographic technologies. For instance, \cite{Gopala2008} the authors consider the secure transmission of information over an ergodic fading channel in the presence of an eavesdropper. An extension of this work considering a multiple-input multiple-output (MIMO) wiretap channel is analyzed in \cite{Oggier2011}. In \cite{Liu2013}, an analysis is conducted on the probability of secrecy capacity for wireless communications over the Rician fading channels. The communication between two legitimate peers in the presence of an external eavesdropper in the context of free-space optical (FSO) communications is analyzed in \cite{Lopez2015}. In \cite{Chen2017}, a comprehensive survey on various multiple-antenna techniques in physical layer security is provided, with an emphasis on transmit beamforming designs for multiple-antenna nodes. An overview on the state-of-the-art works on physical layer security technologies that can provide secure communications in wireless systems is given in \cite{Liu2017-2}. 

In the particular field of physical layer security with NOMA, a small number of contributions are available. { A simple scenario for a DL NOMA with just one eavesdropper (SISO antenna configuration) in a single cell is addressed in \cite{He2017}.} An analysis of the optimal power allocation policy that maximizes the secrecy sum rate for a DL NOMA scenario is presented in \cite{Zhang2016}. { In \cite{Chen18}, a cooperative NOMA system with a single relay is analyzed assuming that the NOMA users are affected by an eavesdropper.} The work in \cite{Qin2016} analyzes the secrecy outage probability (SOP) in a single-cell DL NOMA scenario in which the eavesdroppers are not part of the cellular system.  \cite{Liu2017} extends previous work by proposing several mechanisms to enhance the SOP in a DL NOMA multi-antenna aided transmission. {  In \cite{Lv2018}, a downlink NOMA scenario with multiple-input single-output (MISO) is addressed, proposing a secure beamforming transmission scheme. The secrecy performance of a two-user downlink NOMA with transmit antenna selection schemes is analyzed in \cite{Lei2017}. The work in \cite{Lei2018} studies the secrecy performance of a dowlink of multiple-input multiple-output (MIMO) scenario, focusing on the impact of a max-min transmit antenna selection strategy. Very recently, one work addressing physical layer security in UL NOMA was published \cite{Jiang18}, although it does not make use of any stochastic geometry tool since locations are deterministic.}


\subsection{Motivation and Contributions}

{ The main technical differences and challenges on analyzing the physical layer security in uplink NOMA from the existing studies for downlink NOMA are the following: 
\begin{itemize}
\item
In the uplink NOMA, the BS receives transmissions from all users simultaneously, and consequently, intra-cell interference to a given user is a function of the channel statistics of other users within the cell; however, in downlink NOMA, the intra-cell interference to a user is a function of its own channel statistics \cite{Tabassum17}.
\item
Intra-cluster interfering signals in the uplink NOMA are also the desired signals, therefore, it is not possible to provide the benefits of SIC (enhance the SINR) unequivocally for all users.
\item
In the uplink NOMA, eavesdroppers are randomly positioned near the $N$ legitimate transmitters, independently of the transmitters' location within the cell, whereas in the downlink NOMA, the base station is the unique transmitter, thus simplifying the scenario.
\end{itemize}
}

In this work, we provide a characterization of the physical layer security of UL NOMA. In particular, we provide the following contributions:

\begin{enumerate}
\item
{ We provide new analytical expressions for UL NOMA at the base station with multiple antennas, random spatial locations of eavesdroppers and a protection radius around the legitimate users. This scenario has not been addressed yet to the best of the authors knowledge.} We consider a protected zone around the LUs to establish an eavesdropper exclusion area. 

\item
We analyze the effective secrecy throughput (EST) \cite{Yan15} for uplink NOMA as a performance metric that captures the two key features of wiretap channels (reliability and secrecy) for any number of legitimate users. 

\item
We analyze previous metrics under two different scenarios: fixed and adaptive transmission schemes from LUs. In the case of fixed transmission rate, the impact of assuming a perfect or imperfect SIC is studied. Our analysis allows determining the wiretap code rates that achieve the locally maximum EST for both scenarios.
 
\end{enumerate}

\subsection{Organization and Notation}
The remainder of this paper is organized as follows. The system model under analysis is introduced in Section II. The analysis of the Signal-to-Interference plus Noise Ratio (SINR) distributions for both legitimate users and eavesdroppers is presented in Section III. In Section IV, analytical expressions for the EST under different scenarios are derived. Numerical results are shown and described in Section V. Finally, we draw conclusions in Section VI.

\textbf{Notation}: Throughout this paper, $\mathbb{E}[\cdot]$ stands for the expectation operator and $\mathbb{P}$ for the probability measure. Random variables (RV) are represented with capital letters whereas lower case is reserved for deterministic values and parameters. If $X$ is a RV, $f_X(\cdot)$, $F_X(\cdot)$, $\bar{F}_X(\cdot)$ and $\mathcal{L}_X(\cdot)$ represent its probability density function (pdf), cumulative distribution function (cdf), complementary cdf (ccdf) and Laplace transform of its pdf, respectively.

\section{System Model}
\label{System Model}
We focus on the UL communication scenario in which LUs are connected to a base station (BS) of radius $r_c$ and centered at the origin. We assume a single cell scenario, as considered in most previous studies related to NOMA \cite{Ding14, Qin2016, Yang16, Zhang2016, Zhang16b, Ding17, Liu2017, Rabee17}. A number of eavesdroppers (EDs) are { randomly} distributed along the whole plane, attempting to intercept the communication between LUs and BS. The spatial distribution of EDs is modeled using a homogeneous Poisson Point Process (PPP) uniformly distributed in $\mathbb{R}^2$, which is denoted by $\Phi_e$ and associated with a density $\lambda_e$. An eavesdropper-exclusion zone of radius $r_p$ (in which no eavesdroppers are allowed to roam) is introduced around the LUs for improving the secrecy performance, as it is also considered in \cite{Liu2017} for the downlink. Fig. \ref{fig1} shows the system model under analysis.

\begin{figure}[h]
\centering
\includegraphics[width=3.4in]{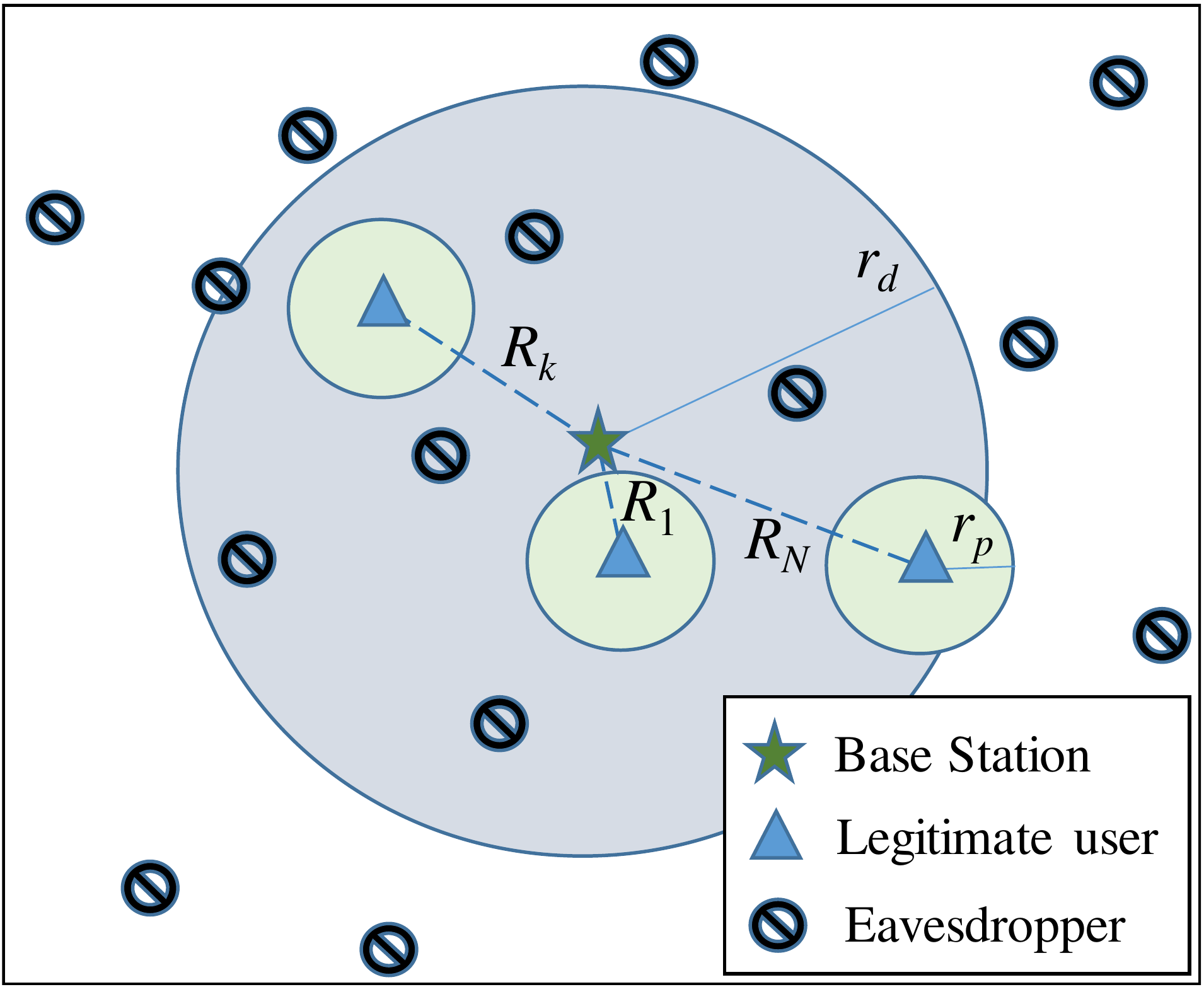}
\caption{System model for secure transmission in UL NOMA}
\label{fig1}
\end{figure}

At each radio resource, the BS gives service to $N$ simultaneous LUs (using NOMA), whose positions are random inside the cell. We assume a random scheduling, i.e. the BS selects randomly the set of $N$ LUs to be scheduled in a given radio resource according to NOMA. The locations of the LUs that are scheduled in a single radio resource are assumed to be uniformly distributed in the cell. Hence, we consider that the resulting set of points (LUs) inside the disk $B(0,r_c)$ is a Binomial Point Process (BPP) $\Phi_B$ with $N$ points, as it is normally assumed in the literature \cite{Zhang16b,Ding14}. { The assumption of a BPP for LUs (instead of a PPP) is due to tractability issues, but at the same time, it provides the necessary spatial correlation between the nodes that are served by the BS. }

We assume that both legitimate transmitters (LUs) and eavesdroppers are equipped with a single antenna each { whereas the BS is equipped with $M$ uncorrelated receive antennas and applies a Maximal Ratio Combining (MRC) reception}. We also assume that LUs' channels and EDs' channels are subject to independent quasi-static Rayleigh fading with equal block length. UL transmit power control is not recommended as justified in the introduction section, and hence, it is not used. 

{ We also consider that the eavesdroppers apply the same SIC method than the BS in order to separate the signals from each transmitter, so that they can technically compromise the communication from each user in the network (and specially if colluding eavesdroppers are considered, which is not the case in this paper). That is, the eavesdroppers measure the received signal power and  first decodes the strongest signal by treating other signals as noise. Afterwards, it cancels the first decoded signal and continues decoding the second strongest signal, and so on. We consider the most detrimental eavesdropper scenario in order to simplify the code rate design.  } 

As stated in \cite{Tabassum16}, the impact of the path-loss factor is generally more dominant than channel fading effects. Hence, for tractability reasons, we assume that ordering of the received signal powers can be approximately achieved by ordering the distances of the users to their serving BS. Let $R_{k}$ be the distance between the $kth$ user and the BS, being $R_{1}\le R_{k} \le  R_{N}$. Power loss due to propagation is modeled using a standard path loss model with $\alpha > 2$, whereas a Rayleigh model is assumed for small-scale fading. Hence, the received signal power at a distance $R_k$ can be simply computed as $H_k {R_k}^{-\alpha}$, where $H_k$ is the fading coefficient. { Note that, since we consider a MRC reception technique at the BS, the desired signal is given by the sum of $M$ independent unit-mean exponentially distributed random variables, yielding a Gamma distribution with ccdf given by
\begin{align}
\label{gammaccdf}
{\overline F _{H_k } (M,x) = e^{ - x }  \cdot \sum\limits_{r = 0}^{M - 1} {\frac{ {x} ^r}{{r!}} } }
\end{align}}
We consider an scenario in which EDs are not a part of the cellular system (passive eavesdropping) and therefore, the channel state information (CSI) associated with EDs' channel is not available at the base station. In addition, we address two different cases regarding the LUs transmission mode:
\begin{itemize}
\item
\textit{Fixed transmission rate}: LUs transmit their information towards their BS at a fixed rate. In this scenario, we find the optimum values for the wiretap code rates, taking into account the reliability outage probability that occurs when the selected fixed rate exceeds the instantaneous channel capacity.
\item
\textit{Adaptive transmission rate}: the BS enforces an adaptive secure transmission from LUs assuming a perfect channel estimation. In this scenario, we find the optimum value of the redundancy rate, $R_e$, that maximizes the secrecy performance.
\end{itemize}

\section{Analysis of the SINR distributions}
First, we analyze the connection related statistics of this scenario using a stochastic geometry approach. We assume that the BS applies SIC to detect the UL transmission from the nearest user first, and afterwards, it continues decoding the information from other users up to user $N$. The received instantaneous SINR at the BS of the $kth$ user can be written as:
\begin{equation}
\label{gammak}
{\gamma _k} = \frac{{{H_k}R_k^{ - \alpha }}}{{I  + {1 \mathord{\left/
 {\vphantom {1 \rho_b }} \right.
 \kern-\nulldelimiterspace} \rho_b }}}
\end{equation}
\noindent
where $I=\sum\nolimits_{j = k + 1}^N {{H_j}R_j^{ - \alpha }}$ represents the intra-cluster interference due to other NOMA users; $\rho_b$ represents the transmit signal-to-noise ratio (SNR) defined as $\rho_b  = \frac{{{P_T}}}{{{\sigma _b^2}}}$, being $P_T$ transmit power at the user terminal and $\sigma_b^2$ the additive white Gaussian noise (AWGN) power received at the BS. Note that (\ref{gammak}) represents the SINR associated with the decoding process of the message from user $k$ subject to the correct decoding process from previous NOMA users (from user 1 to $k-1$) so that their intra-cluster interference has been successfully canceled. Also note that the SINR expression for the last user is simplified to ${\gamma _N} = \rho_b H_N{R_N^{ - \alpha }}$ since the intra-cluster interference has been completely canceled.  

\subsection{Distribution of the SINR of Legitimate Users}
In this section we compute the coverage probability of the legitimate users, i.e. the complementary cumulative distribution function (ccdf) of their received SINR at the BS, which represents the probability for a user to have a SINR higher than a given threshold $t$. 

\begin{lemma}
{ In the case of $M$ antennas at the BS, the ccdf of the SINR for the $kth$ user, $p_k(t)$, is given by
\begin{align}
\label{pk_MRC_final}
\bar{F}_{\gamma_k}(t) &=
 \int_0^{{r_c}} {{{\rm{e}}^{ - {{\psi} \mathord{\left/
 {\vphantom {{\psi} {{\rho _b}}}} \right.
 \kern-\nulldelimiterspace} {{\rho _b}}}}}} \sum\limits_{r = 0}^{M - 1} {\sum\limits_{k = 0}^r {\frac{{{{ \psi }^r}{{\left( { - 1} \right)}^{k}}}}{{(r - k)!k!\rho _b^{r - k}}}} } \frac{{{{\rm{d}}^k}}}{{{\rm{d}}{s^k}}}{{\cal L}_{I|{r_k}}}\left( s \right)|_{s=\psi} \nonumber \\ 
 & \times\frac{2}{{{r_c}}}\frac{{\Gamma \left( {k + \frac{1}{2}} \right)\Gamma \left( {N + 1} \right)}}{{\Gamma \left( k \right)\Gamma \left( {N + \frac{3}{2}} \right)}}\beta \left( {\frac{{r_k^2}}{{r_c^2}};k + \frac{1}{2},N - k + 1} \right){\rm{d}}{r_k}
\end{align}
\noindent
with $\psi=t r_k^\alpha$ and
\begin{align}
\label{Li}
\mathcal{L}_{I|{r_k}}&(s) = {\left( {\frac{{2\left( {r_c^{\alpha + 2}{{}_2{F_1}}\left[ {1,\frac{{\alpha + 2}}{\alpha },2 + \frac{2}{\alpha },-\frac{r_c^\alpha}{t r_k^\alpha}
} \right] - r_k^{\alpha + 2}{{}_2{F_1}}\left[ {1,\frac{{\alpha + 2}}{\alpha },2 + \frac{2}{\alpha },-\frac{1}{t} } \right] } \right)}}{{tr_k^\alpha \left( {r_c^2 - r_k^2} \right)\left( {\alpha + 2} \right)}}} \right)^{N - k}} 
\end{align}}
\noindent where $_2F_1(\cdot,\cdot,\cdot,\cdot)$ is the Gauss hypergeometric function defined in \cite{Abramowitz65} (Ch. 15), $\Gamma(z)=\int_{0}^{\infty} t^{z-1} \mathrm{e}^{-t} \mathrm{d}t$ stands for the Euler Gamma function, $\beta(x; a, b)$ is the beta density function defined as $\beta(x; a, b) = (1/B(a, b))x^{a-1}(1-x)^{b-1}$, being $B(a, b)$ the beta function, which is expressible in terms of Gamma functions as $B(a, b) = \Gamma(a)\Gamma(b)/\Gamma(a + b)$. Note that (\ref{pk_MRC_final}) just includes one finite integral, which can be also computed by the Gaussian-Chebyshev quadrature relationship \cite{Hildebrand87}. 
\end{lemma}

\begin{proof}
See Appendix \ref{Appendix_pk}.
\end{proof}

\begin{corollary}
{ In the case of single antenna $(M=1)$ at the BS, the ccdf of the SINR for the $kth$ user, $p_k(t)$, is simplified to}
\begin{align}
\label{pkfinal}
\bar{F}_{\gamma_k}(t) = \int_0^{{r_c}} &{{{\rm{e}}^{ - tr_k^\alpha /\rho_b }}} 
{\left( {\frac{{2\left( {r_c^{\alpha + 2}{{}_2{F_1}}\left[ {1,\frac{{\alpha + 2}}{\alpha },2 + \frac{2}{\alpha },-\frac{r_c^\alpha}{t r_k^\alpha}
} \right] - r_k^{\alpha + 2}{{}_2{F_1}}\left[ {1,\frac{{\alpha + 2}}{\alpha },2 + \frac{2}{\alpha },-\frac{1}{t} } \right] } \right)}}{{tr_k^\alpha \left( {r_c^2 - r_k^2} \right)\left( {\alpha + 2} \right)}}} \right)^{N - k}}
\nonumber \\ & \times
\frac{2}{{{r_c}}}\frac{{\Gamma \left( {k + \frac{1}{2}} \right)\Gamma \left( {N + 1} \right)}}{{\Gamma \left( k \right)\Gamma \left( {N + \frac{3}{2}} \right)}}\beta \left( {\frac{{r_k^2}}{{r_c^2}};k + \frac{1}{2},N - k + 1} \right){\rm{d}}{r_k}
\end{align}
\end{corollary}

\begin{corollary}
{ In the case of single antenna $(M=1)$ at the BS,} the coverage probability for the farthest user ($N$) is simplified to 
\begin{align}
\bar{F}_{\gamma_N}(t)& = \frac{{2N}}{{\alpha r_c^{2N}}}{\left( {\frac{t}{\rho_b }} \right)^{ - \frac{{2N}}{\alpha }}}\left[ {\Gamma \left( {\frac{{2N}}{\alpha }} \right) - \Gamma \left( {\frac{{2N}}{\alpha },\frac{{r_c^\alpha t}}{\rho_b }} \right)} \right]
\end{align}
where $\Gamma(\cdot,\cdot)$ stands for the upper incomplete Gamma function.
\end{corollary}

\begin{proof}
The farthest user ($N$) experiences no intra-cluster interference, so its coverage probability can be expressed as
\begin{align}
\bar{F}_{\gamma_N}(t)& = \int_0^{{r_c}} {{{\rm{e}}^{ - tr_N^\alpha /\rho_b }}} {f_{{R_N}}}({r_N}){\rm{d}}{r_N} \nonumber \\ & = \int_0^{{r_c}} {{{\rm{e}}^{ - tr_N^\alpha /\rho_b }}} \frac{{2N}}{{{r_c}}}\left( {\frac{{r_N^2}}{{r_c^2}}} \right){\rm{d}}{r_N}
\end{align}
After minor manipulations, the proof is complete.
\end{proof}

\subsection{Distribution of the SNR of eavesdroppers}
We address worst-case scenario, in which eavesdroppers are assumed to have strong detection capabilities. Specifically, by applying multi-user detection techniques, the multi-user data stream received at the BS can be also distinguished by the eavesdroppers. 

We consider the most detrimental eavesdropper, which is not necessarily the nearest one, but the one having the best channel to the LU that is transmitting towards the BS. Therefore, the instantaneous received SNR at the most detrimental eavesdropper (with respect with any LU) can be expressed as follows:
\begin{equation}
{\gamma _{{e}}} = \mathop {\max }\limits_{e \in {\Phi _e}} \left\{ {\rho_e}{H_e}{R_e^{-\alpha }} \right\}
\end{equation}
\noindent 
where $\rho_e$ represents the transmit SNR defined as $\rho_e  = \frac{{{P_T}}}{{{\sigma _e^2}}}$, being $P_T$ transmit power at the LU and $\sigma_e^2$ the AWGN power received at the eavesdropper.

\begin{lemma}
Assuming an eavesdropper-exclusion zone or radius $r_p$ around the LUs, the cdf of the SNR for the most detrimental eavesdropper can be computed as follows:
\begin{align}
\label{Fe}
{F_{{\gamma _e}}}&(t)
{ = \exp \left[ { - \frac{{2\pi {\lambda _e}\Gamma \left( {{{2}}/{\alpha },r_p^\alpha t /{\rho_e}} \right)}}{{\alpha {{(t/{\rho_e})}^{2/\alpha }}}}} \right]}
\end{align}
\end{lemma}

\begin{proof}
Taking into account that EDs follow a PPP distribution, we can express the cdf of the SNR for the most detrimental eavesdropper as follows:

 \begin{align}
{F_{{\gamma _e}}}&(t) = 1 - p_e(t)= {E_{{\Phi _e}}}\left\{ {\prod\limits_{e \in {\Phi _e}} {{F_{{H_e}}}\left( {t{r_e^\alpha} /{\rho_e}} \right)} } \right\}
\nonumber \\&
{  \overset{(a)}{=} \exp \left[ { - {\lambda _e}\int_{{R^2}} {\left( {1 - {F_{{H_e}}}\left( {t{r_e^\alpha} / {\rho_e}} \right)} \right){r_e}d{r_e}} } \right]}
\nonumber \\&
{ = \exp \left[ { - 2\pi {\lambda _e}\int_{{r_p}}^\infty  {{r_e}{e^{ - t{r_e^\alpha} /{\rho_e}}}d{r_e}} } \right]} 
\end{align}
\noindent
where $(a)$ comes from the Probability Generating Functional (PGFL) \cite{Chiu13}. Solving the last integral, the proof is complete.
\end{proof}

 In the particular case of no eavesdropper-exclusion zone, (\ref{Fe}) is simplified to:
\begin{equation}
\label{FeRp0}
{\left. {{F_{{\gamma _e}}}(t)} \right|_{{r_p} = 0}} = \exp \left[ { - \frac{{2\pi {\lambda _e}\Gamma \left( {{{2}}/{\alpha }} \right)}}{{\alpha {{(t/{\rho_e})}^{2/\alpha }}}}} \right]
\end{equation}

\section{Secrecy rate metrics}
Let $R_s$ be the secrecy rate in a legitimate link, i.e. the rate of transmitted confidential information. This rate can be computed as:
\begin{equation}
R_s \triangleq R_b - R_e \ge 0,
\end{equation}
\noindent
where $R_b$ represents the codeword rate from the LU to the BS, i.e. rate at which the codeword is transmitted, including the confidential message and redundancy; $R_e$ quantifies the redundancy rate, i.e. rate associated with redundant information for providing physical layer security in the message transmission. Roughly, a larger $R_e$ provides a higher secrecy level. 

On the one hand, if we select a codeword rate such that $R_b \le C_b$ (being $C_b$ the capacity of the legitimate channel), a reliability constraint is ensured. On the other hand, if the redundancy rate is above the capacity of the eavesdropper's channel, i.e. $R_e > C_e$, a secrecy constraint is achieved. 

Depending on whether the CSI of LU and ED links are available at the BS, such rates can be adapted to the channel or not. Most of previous works on physical layer security compute the secrecy capacity as {  ${C_s} = {\left[ {{C_b} - {C_e}} \right]^ + }$, where ${\left[ x \right]^ + } = \max \left\{ {0,x} \right\}$ } \cite{Barros2006}, although this definition implicitly requires that both $C_b$ and $C_e$ are available. In our scenario, this assumption is not realistic since EDs are not part of the cellular system. Subsequently, we do not use the typical information-theoretic formulation related to the secrecy capacity but a recent formulation of a new metric, referred to as the \textit{effective secrecy throughput} (EST) \cite{Yan15}, which captures both the reliability constraint and the secrecy constraint as independent terms. The EST of a wiretap channel quantifies the average secrecy rate at which the messages are transmitted from the LUs to the BS without being leaked to the eavesdroppers, and can be defined as 
\begin{equation}
\label{Phi}
\Phi(R_b, R_e) = (R_b - R_e)\left[1-\O_r(R_b)\right]\left[1-\O_s(R_e)\right]
\end{equation}

\noindent
where the term $(R_b-R_e)$ represents the rate of transmitted confidential information, i.e. $R_s$; and the term $[1-\O_r(R_b)]$ $[1-\O_s(R_e)]$ quantifies the probability that the information is securely transmitted from the LUs to the BSs, being  $[1-\O_r(R_b)]$ associated with the reliability constraint and $[1-\O_r(R_e)]$ associated with the secrecy constraint. We assume a normalized bandwidth $W=1$, and therefore, secrecy rate and capacity metrics are measured in bits/s.

{ We have chosen the EST as a secrecy performance metric in this paper as it allows for explicitly designing the wiretap code rates that satisfy certain reliability and secrecy constraints. This is not the case when using conventional secrecy metrics such as the Secrecy Outage Probability (SOP) ${\P(C_s<R_s)}$ (where $R_s$ is defined as the threshold rate under which secure communication cannot be achieved) or the probability of strictly positive secrecy capacity ${\P(C_s>0)}$. Besides, and despite being a relatively recent performance metric, the EST has been used in numerous recent works \cite{Jiang18,Yan15, Monteiro15,Yu17,Wang16,Yang17,Wang17,Jiang17,Monteiro18,Chen18b}. Additionally, the evaluation of the SOP poses an additional challenge from an analytical perspective in this specific scenario, since it includes an additional infinite integral.}

\subsection{Adaptive Transmission Rate}
In this scenario, the BS enforces an adaptive transmission scheme from LUs in the UL. 

\begin{theorem}
\label{Theo_ESTadapt}
The EST for the NOMA $kth$ user in case of adaptive transmission is given by
{ 
\begin{align}
\label{PhikAdapt}
{\Phi _k}({R_e}) = \left( {\frac{1}{\rm{ln}\,2} \int_{2^{R_e} - 1}^\infty  {\frac{{{{\bar{F}}_{{\gamma _k}}}(z)}}{{1 + z}}{\rm{d}}z}  - {{{{F}}_{{\gamma _k}}}(2^{R_e}-1)}{R_e}} \right){F_{{\gamma _e}}}\left( {{2^{R_e} - 1}} \right)
\end{align}
}
\noindent
where $\bar{F}_{\gamma _k}(\cdot)$ and ${F}_{\gamma _e}(\cdot)$ were given in (\ref{pkfinal}) and (\ref{Fe}), respectively.
\end{theorem}

\begin{proof}
In case of adaptive transmission, $R_b$ can be optimally chosen as \mbox{$R_b = C_b$}, and hence, the reliability constraint can be always guaranteed, i.e. the reliability outage probability is zero: $\O_r(R_b)=0$. Therefore, the EST for the NOMA $kth$ user can be defined as
\begin{equation}
\label{Phik}
{\Phi _k}(R_e) = ({ C _k} - R_e)\left[1-\O_s(R_e)\right]
\end{equation}
\noindent
where the term $C _k$ represents the ergodic capacity for the $kth$ user. Note that, in the adaptive transmission scheme, $R_e$ is adjusted within the constraint $0 < R_e < C _k$. { Since $R_b=C_b$, we need to guarantee that \mbox{$C_b=log_2(1+\gamma_b)>R_e$}, that is, $\gamma_b>2^{R_e}-1$. Therefore, assuming a normalized channel bandwidth $W=1$, the average capacity for the $kth$ user, $C _k$, can be expressed as
\begin{align}
C _k  = \int_{2^{R_e}-1}^\infty  {{{\log }_2}\left( {1 + \gamma} \right){f_{{\gamma _k}}}(\gamma ){\rm{d}}\gamma}.
\end{align}

Using integration by parts with $u={\log}_2(1+\gamma)$, ${\rm{d}} v = f_{\gamma_k}(\gamma )$ and $v=-\left(1-F_{\gamma _k}(\gamma )\right)$, the average capacity can be also expressed as 
\begin{align}
\label{Ck}
 C _k ={{{\bar{F}}_{{\gamma _k}}}(2^{R_e}-1)}R_e + \frac{1}{\rm{ln}\,2} \int_{2^{R_e}-1}^\infty  {\frac{{{{\bar{F}}_{{\gamma _k}}}(z)}}{{1 + z}}{\rm{d}}z} 
\end{align}
\noindent 

The secrecy outage probability term 
can be computed as
\begin{align}
\label{Os}
\O_s(R_e) &= \mathbb{P}(R_e < C_e) = \mathbb{P}(\gamma_e > 2^{R_e}- 1) = 1- {F}_{\gamma_e}(2^{R_e}- 1).
\end{align}
}
Substituting (\ref{Os}) and (\ref{Ck}) into (\ref{Phik}), the proof is complete.
\end{proof}

\begin{remark}[Impact of eavesdroppers density, $\lambda_e$]
\label{remark_density}
In view of Theorem \ref{Theo_ESTadapt}, it can be deduced that, for $\lambda_e=0$, the term associated with the secrecy constraint, $\O_r(R_e)$, is null; hence, the EST is mainly determined by the capacity of the LU's link. On the other hand, the EST tends to zero as $\lambda_e$ grows since expression (\ref{PhikAdapt}) always satisfies that $\mathop {\lim }\limits_{{\lambda _e} \to \infty } {\Phi _k}\left( {{R_e}} \right) = 0,\; \forall {r_p} \in [0,\infty )$; { this is due to the fact that the fading distribution introduces a non-null probability of having a higher instantaneous capacity for the eavesdropper than for the legitimate user.}
\end{remark}

\begin{remark}[Impact of eavesdropper-exclusion radius, $r_p$]
\label{remark_rp}
In view of expression (\ref{PhikAdapt}), it can be noted that  the only term that depends on $r_p$ is the cdf of the SNR of the worst eavesdropper, ${F_{{\gamma _e}}}\left( {{2^{R_e} - 1}} \right)$; for $r_p=0$, this term is simplified to (\ref{FeRp0}), whereas for $r_p \rightarrow \infty$, this term satisfies that ${F_{{\gamma _e}}}\left( {{2^{R_e} - 1}} \right)|_{r_p \rightarrow \infty}=1$, that is, eavesdroppers do not have any impact on the EST performance.
\end{remark}

\subsection{Fixed Transmission Rate}
In case the LUs use a fixed transmission rate, the reliability constraint cannot be always guaranteed, i.e. a reliability outage must be taken into account as
\begin{align}
\label{Or}
\O_r(R_b) = \mathbb{P}(R_b > C_b) 
\end{align}

Therefore, an outage may occur whenever a message transmission is either unreliable or non secure.

Regarding the reliability constraint term, $\O_r(R_b)$, we address in our analysis the impact of imperfect SIC and detection probability for NOMA. Note that the signals from the intra-cluster interfering users may or may not be decoded perfectly; therefore, SIC may or may not be performed in a perfect fashion. As a consequence, we distinguish two cases: perfect and imperfect SIC.

The reliability constraint term for user $k$ in the case of perfect SIC, named as $p_k^{(P)}$, is given by 
\begin{align}
p_k^{(P)}(R_b) &= 1 - {\O_{{r_k}}}({R_b}){\rm{ }} = 1 - \mathbb{P}({R_b} > {C_k}) \nonumber \\ & = 1 - \mathbb{P}({\gamma _k} < {2^{R_b}} - 1) = {\bar{F}_{{\gamma _k}}}({2^{{R_b}}} - 1)
\end{align}
That is, the reliability constraint term represents the detection probability for user $k$, whose expression was obtained in (\ref{pkfinal}).

Finally, the EST for user $k$ in case of perfect SIC can be expressed as
\begin{align}
\Phi_k^{(P)} ({R_b},{R_e}) = ({R_b} - {R_e}){\bar{F}_{{\gamma _k}}}\left( {{2^{R_b} - 1}} \right){F_{{\gamma _e}}}\left( {{2^{R_e} - 1}} \right)
\end{align}

However, in the case of imperfect SIC, the intra-cluster interference experienced by the $kth$ user depends on whether the detection for the $k-1$ nearest users were successful or not, which complicates the model significantly. In this paper we assume the worst case of imperfect SIC, which considers that the decoding of the $kth$ user is always unsuccessful whenever the decoding of his relative $k-1$ closest users is unsuccessful \citep{Tabassum17}. Therefore, the reliability constraint term for the worst-case detection probability of $kth$ user is given by:
\begin{equation}
p_k^{(I)}(R_b) = \prod\limits_{i = 1}^{k} \bar{F}_{\gamma_i}\left(2^{R_b} -1\right)
\label{pkI}
\end{equation}

Finally, the EST for user $k$ in case of imperfect SIC can be expressed as
\begin{align}
\Phi_k^{(I)} ({R_b},{R_e}) = ({R_b} - {R_e})\prod\limits_{i = 1}^{k} \bar{F}_{\gamma_i}\left(2^{R_b} -1\right){F_{{\gamma _e}}}\left( {{2^{R_e} - 1}} \right)
\end{align}

\begin{table*} 
\centering 
\caption{Summary of secrecy metric expressions for different scenarios}
\renewcommand{\arraystretch}{1.4}
\begin{tabular}{|p{1.9cm}|p{2.2cm}|p{1.85cm}|p{9.1cm}|}
\hline 
\textsc{Scenario}&\textsc{Reliability constraint, $\left[1-\O_r(R_b)\right]$}&\textsc{Secrecy constraint, $\left[1-\O_s(R_e)\right]$}&\textsc{Effective Secrecy Throughput (EST), $\Phi_k$}\\[0.5ex] 
\hline 
Adaptive rate& 1 &${F_{{\gamma _e}}}\left( {{2^{R_e} - 1}} \right)$&${ {\Phi _k}({R_e}) = \left( {\frac{1}{\rm{ln}\,2} \int_{2^{R_e} - 1}^\infty  {\frac{{{{\bar{F}}_{{\gamma _k}}}(z)}}{{1 + z}}{\rm{d}}z}  - {{{{F}}_{{\gamma _k}}}(2^{R_e}-1)}{R_e}} \right){F_{{\gamma _e}}}\left( {{2^{R_e} - 1}} \right)}$
\\ [4pt] 
\hline 
Fixed rate with perfect SIC& ${\bar{F}_{{\gamma _k}}}({2^{{R_b}}} - 1)$ &${F_{{\gamma _e}}}\left( {{2^{R_e} - 1}} \right)$&$\Phi_k^{(P)} ({R_b},{R_e}) = ({R_b} - {R_e}){\bar{F}_{{\gamma _k}}}\left( {{2^{R_b} - 1}} \right){F_{{\gamma _e}}}\left( {{2^{R_e} - 1}} \right)$\\[4pt] 
\hline
Fixed rate with imperfect SIC& $\prod\limits_{i = 1}^{k} \bar{F}_{\gamma_i}\left(2^{R_b} -1\right)$&${F_{{\gamma _e}}}\left( {{2^{R_e} - 1}} \right)$&$\Phi_k^{(I)} ({R_b},{R_e}) = ({R_b} - {R_e})\prod\limits_{i = 1}^{k} \bar{F}_{\gamma_i}\left(2^{R_b} -1\right){F_{{\gamma _e}}}\left( {{2^{R_e} - 1}} \right)$\\ [4pt] 
\hline
\end{tabular}
\label{table1}
\end{table*}

Note that the EST expression is the same for the first NOMA user independently of the SIC assumption, i.e. $\Phi_1^{(P)}=\Phi_1^{(I)}$, since potential detection errors occur from the second user up to the $Nth$ user. 

A summary of secrecy metric expressions for different scenarios
in shown in Table \ref{table1}.

\section{Numerical Results}
In this section, analytical results are illustrated and validated with extensive Monte Carlo simulations in order to assess the physical layer security in UL NOMA. We conduct a thorough performance comparison between the adaptive and fixed rate transmission schemes in terms of EST. Main parameters are presented in Table I unless otherwise stated. 

\begin{table}
\caption{Main configuration parameters} 
\label{params}
\centering
\begin{tabular}{ c c }
\hline
Parameter & Value \\
\hline
$r_c$ (m)& 500  \\
$\alpha$ & 3.8  \\
$\rho_b$ (dB)& 110  \\
$\rho_e$ (dB)& 90  \\
$\lambda_e$ (points/m$^2$) & 1e-5  \\
\hline
\end{tabular}
\end{table}

\subsection{Fixed Transmission Rate}
In the case of fixed transmission rate, the reliability constraint term (or equivalently, the detection probability) plays an important role in NOMA performance. Let us analyze first this contribution separately.

Fig. \ref{CovProb_vs_k} shows the detection probability results for legitimate users with perfect SIC, $p_k^{(P)}$, and imperfect SIC, $p_k^{(I)}$. In this case, we have considered a high number of simultaneous NOMA LUs ($N=6$) randomly positioned according to a BPP in order to evaluate the performance as $k$ grows. In the case of perfect SIC, results show that detection probability is not a monotonically decreasing function with $k$ (i.e. with the distance from the $kth$ user to the BS); instead, farthest LUs are boosted since the intra-cluster interference term has been partially (or totally) canceled. Note that the best result is achieved for the farthest user, $k=N=6$, since perfect SIC assumes that intra-cluster interference is fully canceled. However, in the case of imperfect SIC, the intra-cluster interference experienced by the $kth$ user depends on whether the detection for $k-1$ nearer users were successful or not, thus providing a monotonically decreasing function with $k$. Note also that higher values of $R_b$ lead to a lower detection probability.

\begin{figure}[!ht]
\begin{center}
\includegraphics[width=0.7\columnwidth]{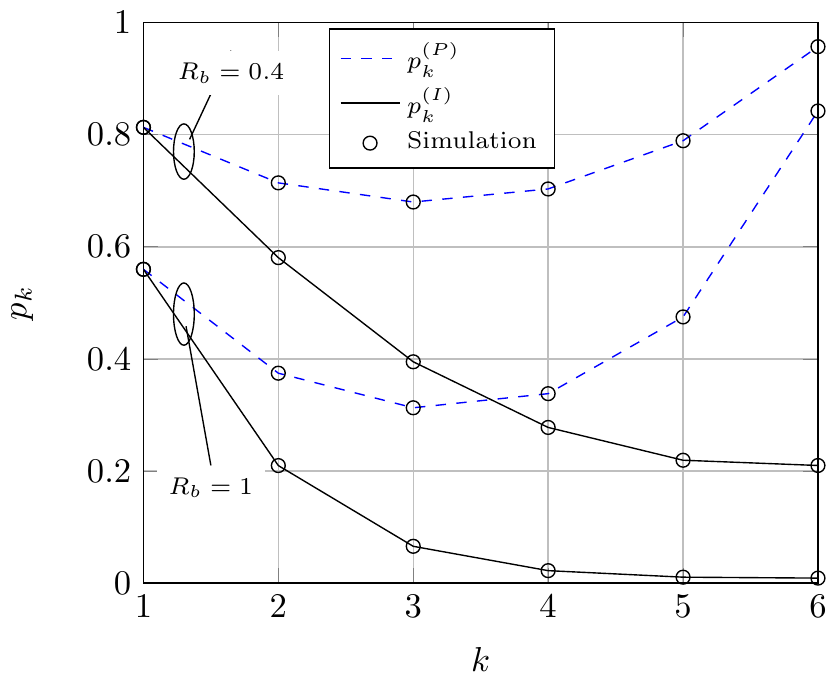}
\caption{Detection probability for LUs with perfect and imperfect SIC for each user $k$ with $N=6$.}
\label{CovProb_vs_k}
\end{center}
\end{figure}

Fig. \ref{CovProb_vs_t} shows the detection probability results for legitimate users with imperfect SIC, $p_k^{(I)}$, and for eavesdroppers, $p_e$, versus the SINR threshold $t=2^{R_b}-1$. Results for $p_k^{(I)}$ are obtained from (\ref{pkI}) considering $N=4$ NOMA LUs. The detection probability of eavesdroppers, $p_e$, is also shown for different values of the exclusion area radius, $r_p$. Since we consider the most detrimental eavesdropper, i.e. the one receiving the best channel quality from the LU, the detection probability results for the eavesdroppers may outperform the results for LUs as $r_p$ is decreased, assuming a density of eavesdroppers of $\lambda_e=1$ (default value). This undesirable scenario can be compensated by increasing the exclusion area radius. { . Results also show the detection probability for legitimate users in case of different number of antennas at the BS, leading to an important improvement as $M$ grows.}

\begin{figure}[!ht]
\begin{center}
\includegraphics[width=0.7\columnwidth]{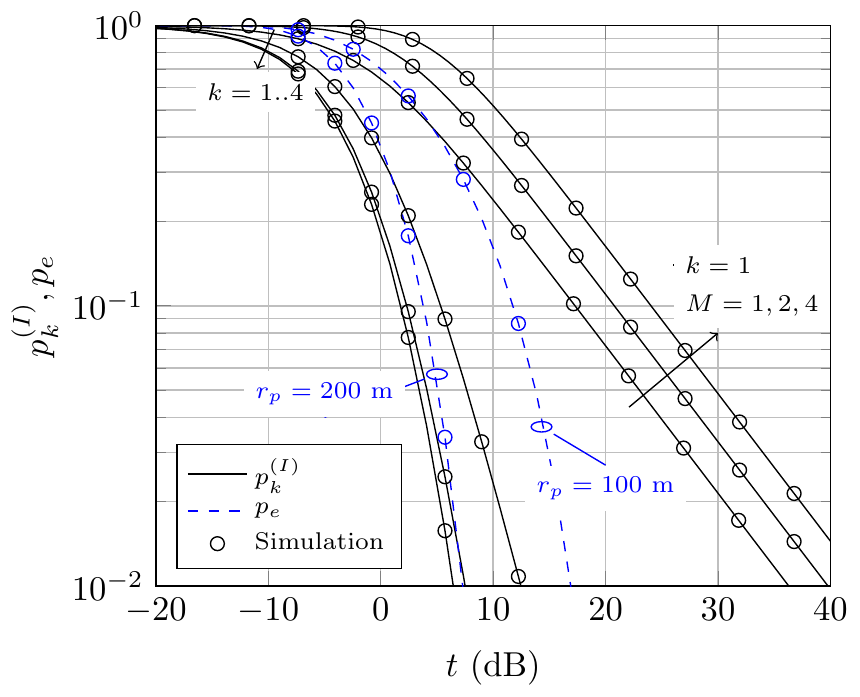}
\caption{Detection probability for LUs with imperfect SIC, $p_k^{(I)}$, and for eavesdroppers, $p_e$, versus $t=2^{R_b}-1$, with $N=4$ {  and different number of antennas $M$.}}
\label{CovProb_vs_t}
\end{center}
\end{figure}

Fig. \ref{ESTfixedN2k1} shows the EST for fixed rate transmission scheme and perfect SIC versus $R_b$ and $R_e$, $\Phi_k^{(P)}(R_b, R_e)$. We observe that there is a unique pair of $R_b$ and $R_e$ that maximizes the EST. Also note that EST is null for $R_e \le R_b$.

\begin{figure}[!ht]
\centering
\includegraphics[width=0.8\columnwidth]{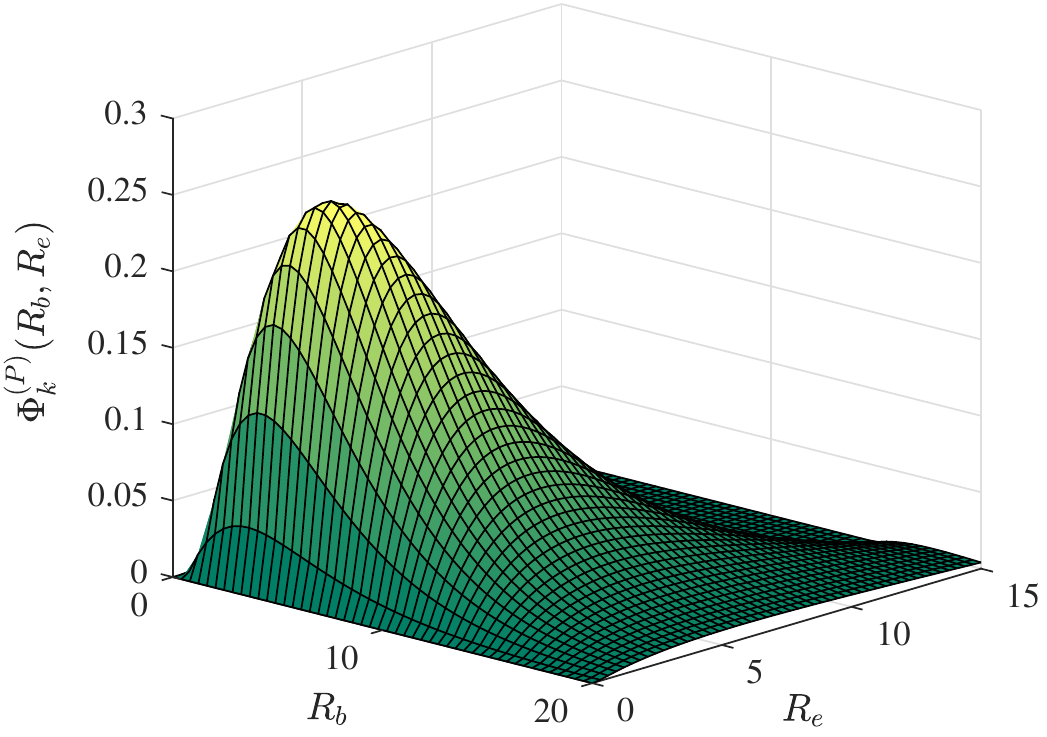}
\caption{EST with perfect SIC for fixed transmission rate with $N=2$, $k=1$ and $r_p=50$ m.}
\label{ESTfixedN2k1}
\end{figure}

The value of $R_e$ that maximizes the EST, noted as $R_e^\dag$, has been determined numerically and shown in Fig. \ref{graficaESTMaxRe} as a function of $R_b$ and $\lambda_e$, with $N=2$, $k=1$ and $r_p=50$ m. Note that the ratio between $R_e^\dag$ and $R_b$ is not linear. We also observe that a higher density of eavesdroppers requires a higher redundancy rate to optimize the EST.

\begin{figure}[!ht]
\centering
\includegraphics[width=0.7\columnwidth]{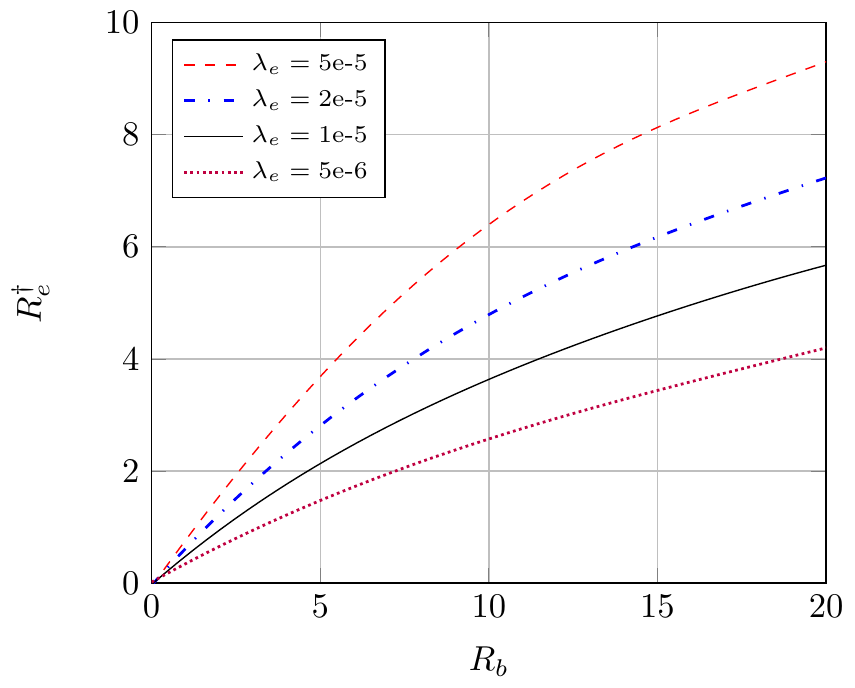}
\caption{Optimum value of $R_e$ that maximizes the EST as a function of $R_b$ and $\lambda_e$, considering fixed transmission with $N=2$, $k=1$ and $r_p=50$ m.}
\label{graficaESTMaxRe}
\end{figure}

Fig. \ref{graficaESTfixedvsRb} shows a comparison between the EST for fixed rate transmission with perfect SIC, $\Phi_k^{(P)}$, and imperfect SIC, $\Phi_k^{(I)}$. EST results are shown for $N=2$ NOMA users as a function of $R_b$, assuming a value of $R_e=3$ bps and $r_p=50$ m. We observe that the results for the first user ($k=1$) are the same for perfect and imperfect SIC since imperfect SIC models the propagation of decoding errors from previous decoded users. We also observe that, in the case of perfect SIC, the maximum EST for the second user is not degraded significantly compared to the first user, as the larger distance to the BS is compensated by the fact that the second user does not experience (ideally) any intra-cluster interference. However, in the case of imperfect SIC, the second user is highly degraded compared to the first user due to SIC error propagation from the previous decoded user. Note also that the value of $R_b$ that maximizes the EST is different of each LU, so optimum code rate selection at the base station must be done per LU.

\begin{figure}[!ht]
\centering
\includegraphics[width=0.7\columnwidth]{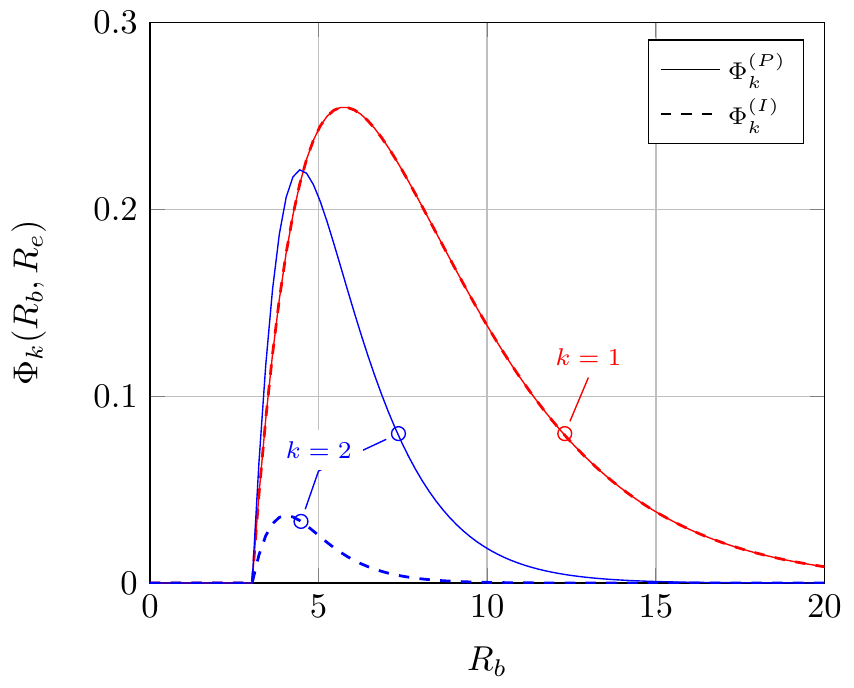}
\caption{Comparison between the EST for fixed rate transmission scheme with perfect SIC and imperfect SIC, with $N=2$, $R_e=3$ and $r_p=50$ m.}
\label{graficaESTfixedvsRb}
\end{figure}

{ Fig. \ref{figRpmin} shows the value of the minimum eavesdropper-exclusion radius ($r_{p_{min}}$) that ensures a certain EST value. Results are shown for the first user ($k=1$), with $N=2$, as a function of the eavesdropper density ($\lambda_e$). It is observed that a higher eavesdropper-exclusion radius is require to achieve the minimum EST target as $\lambda_e$ or $\rho_e$ is increased. It is also observed that for low $\lambda_e$ values, there is no need of including an exclusion area to achieve the EST target.}

\begin{figure}[!ht]
\begin{center}
\includegraphics[width=0.7\columnwidth]{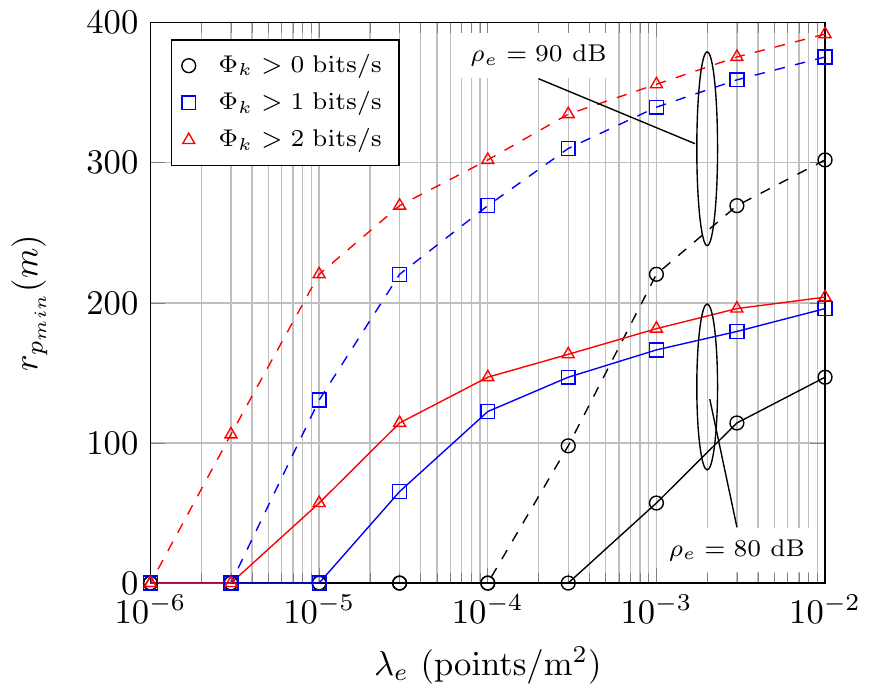}
\caption{{ Minimum value of the eavesdropper-exclusion radius ($r_{p_{min}}$) that ensures a target EST ($\Phi _k$) as a function of $\lambda_e$, for $N=2$, $k=1$ and $R_e=1$.}}
\label{figRpmin}
\end{center}
\end{figure}

\subsection{Adaptive Transmission Rate}

In this section we provide performance results in case the BS uses the CSI of LUs to enforce an adaptive transmission scheme. 

Fig. \ref{EST_vs_Re} shows EST results of the first user ($k=1$) as a function of the redundancy rate, $R_e$, assuming $N=2$ NOMA users. In this case, $R_b$ is adapted to the channel capacity, i.e. $R_b=C_b$, whereas the value of $R_e$ must be properly designed. In that sense, there is a value of $R_e$ that maximizes the EST. We also observe that higher eavesdropper-exclusion radii enhance the EST. As mentioned before, in case of adaptive transmission, the reliability constraint does not affect the secrecy performance; therefore, no SIC errors are considered in this case. Note that if the difference between $\rho_b$ and $\rho_e$ is higher (due to the value of $\sigma_e^2$ compared to $\sigma_b^2$), the EST is considerably increased.

\begin{figure}[!ht]
\begin{center}
\includegraphics[width=0.7\columnwidth]{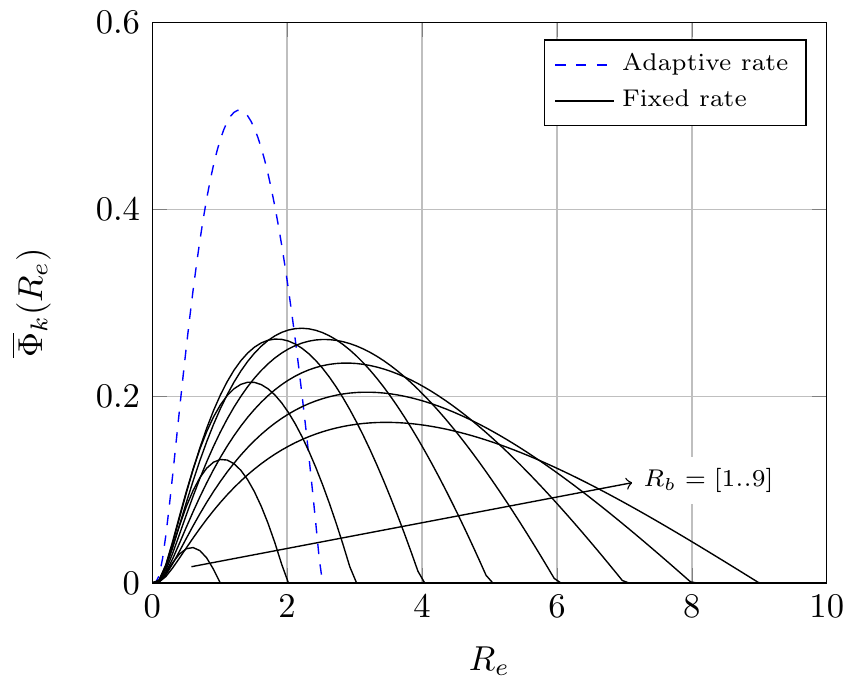}
\caption{{ Comparison between the EST for adaptive rate and fixed rate transmission schemes with perfect SIC, with $N = 2$, $k=1$, $r_p = 50$ m.}}
\label{Fixed_vs_adaptive_EST}
\end{center}
\end{figure}

{ Fig. \ref{Fixed_vs_adaptive_EST} shows a performance comparison between the EST for adaptive rate and fixed rate transmission schemes with perfect SIC as a function of $R_e$, with $N = 2$, $k=1$ and $r_p = 50$ m. Fixed rate results are plotted for different $R_b$ values (from 1 to 9) whereas adaptive rate transmission scheme selects a value of $R_b$ such that $R_b = C_b$. It is observed that adaptive rate transmission outperforms fixed rate transmission. The reason is that, in case of adaptive transmission, the reliability constraint does not affect the performance as $\left[1-\O_r(R_b)\right]=1$, whereas in case of fixed transmission with perfect SIC: $\left[1-\O_r(R_b)\right]={F_{{\gamma _e}}}\left( {{2^{R_e} - 1}} \right)$. It is also observed that the value of $R_e$ than maximizes the EST in both cases are different.}

\begin{figure}[!ht]
\begin{center}
\includegraphics[width=0.7\columnwidth]{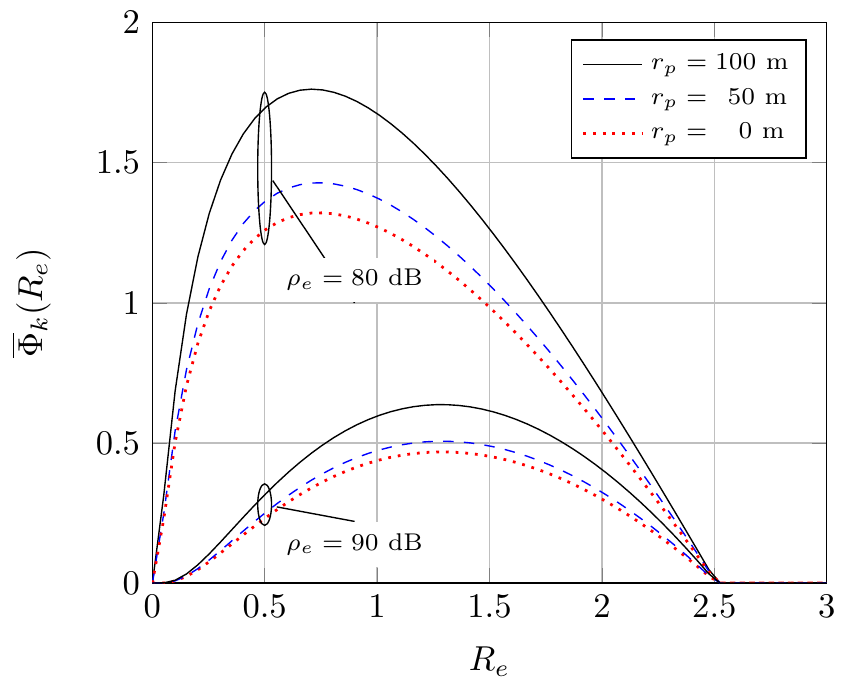}
\caption{EST of the first user ($k=1$) versus $R_e$ for adaptive transmission with $N=2$ and $\rho_b=110$ dB.}
\label{EST_vs_Re}
\end{center}
\end{figure}

The impact of the eavesdropper-exclusion radius on the EST is depicted in Fig. \ref{graficaESTadaptvsRp}. We observe an increasing S-shape behavior as $r_p$ grows, since the most detrimental eavesdropper reduces its detection capabilities for higher $r_p$ values. Results match perfectly with Remark \ref{remark_rp}, which stated that for $r_p \rightarrow \infty$,  eavesdroppers do not have any impact on the performance.

\begin{figure}[!ht]
\centering
\includegraphics[width=0.7\columnwidth]{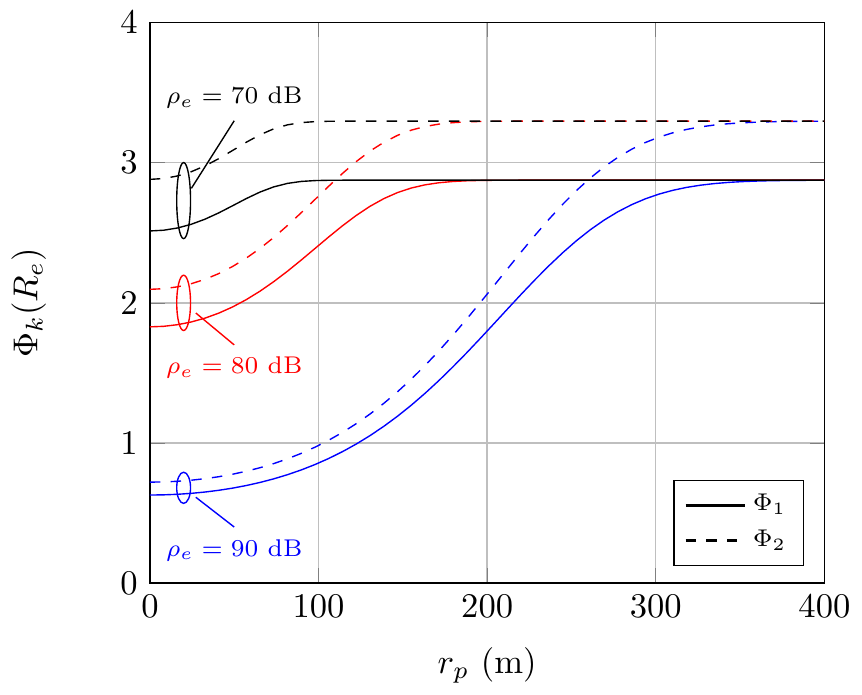}
\caption{EST of the $kth$ user versus $r_p$ for adaptive transmission with $N=2$ and $\rho_b=110$ dB.}
\label{graficaESTadaptvsRp}
\end{figure}

EST results as a function of the density of eavesdroppers, $\lambda_e$, is shown in Fig. \ref{graficaESTadaptvsLambda}. We observe an exponential decreasing behavior with $\lambda_e$. As stated in Remark \ref{remark_density}, when $\lambda_e$ tends to zero, the EST is mainly determined by the capacity of the LU's link; on the contrary, when $\lambda_e$ tends to infinity, the EST is zero, although higher eavesdropper-exclusion radii lead to a slower EST degradation. {  Results also show the performance gain as the number of antennas $M$ is increased.}

\begin{figure}[!ht]
\centering
\includegraphics[width=0.7\columnwidth]{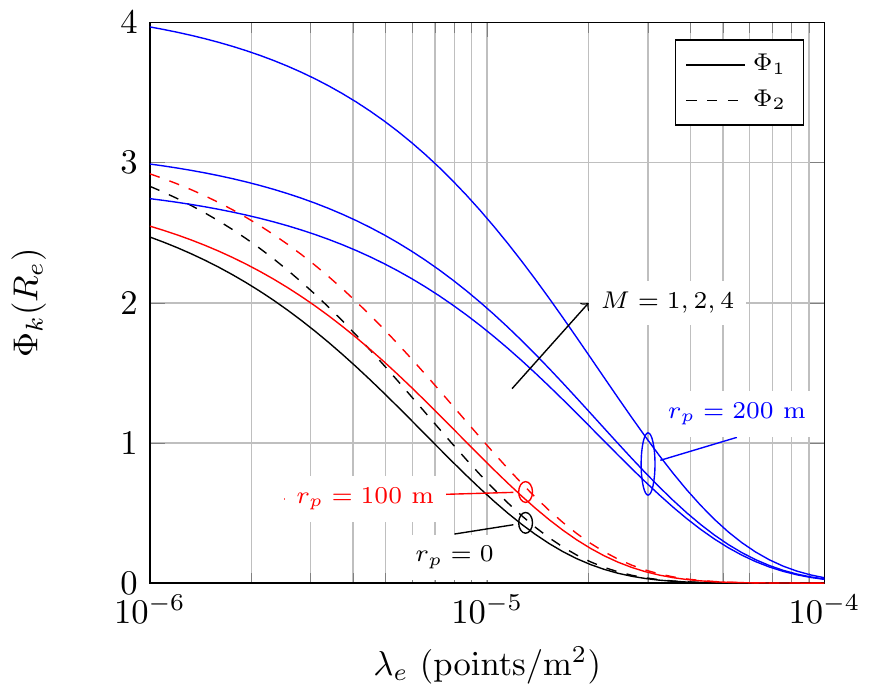}
\caption{EST of the $kth$ user versus $\lambda_e$ for adaptive transmission with $N=2$ and $R_e=1$; {  number of antennas $M=1,2,4$.}}
\label{graficaESTadaptvsLambda}
\end{figure}

Fig. \ref{graficaESTadaptvsPot} shows the EST for adaptive transmission for the $kth$ user as a function of the transmission power $(P_T)$ of the LU measured in dBm/Hz. We observe an optimum value of $P_T$, which depends on the specific values of $\rho_e$ and $k$. We have considered an eavesdropper-exclusion radius of \mbox{$r_p=50$ m} and an average noise power received at the BS of \mbox{$\sigma_b^2=-160$ dBm/Hz}; note that the default value of $\rho_b=110$ dB would give a value of $P_T=-50$ dBm/Hz, or equivalently, a \mbox{$P_T=23$ dBm} for a bandwidth of 20 MHz, which is a typical power value for a micro-cell. Results show that very low $P_T$ values lead to a very poor performance since the average SINR of the LUs is very low (reliability constraint); on the other hand. When then transmit power is increased, there is a optimum value above which the EST starts decreasing, since the eavesdroppers are also increasing their detecting capabilities (secrecy constraint). Results also show that higher values of $\rho_e  = \frac{{{P_T}}}{{{\sigma _e^2}}}$ degrades considerably the EST. We also observe that the performance of the first and second LUs differs significantly as $\rho_e$ is increased. We must recall that in the adaptive transmission, the last user is ideally free of intra-cluster interference, and hence, its performance is limited by noise. Therefore, the second user is much more affected by the value of $\rho_e$. In case of high noise power at eavesdroppers (low $\rho_e$) the second user is shown to outperform the first user despite being further from the BS.  

\begin{figure}[!ht]
\centering
\includegraphics[width=0.7\columnwidth]{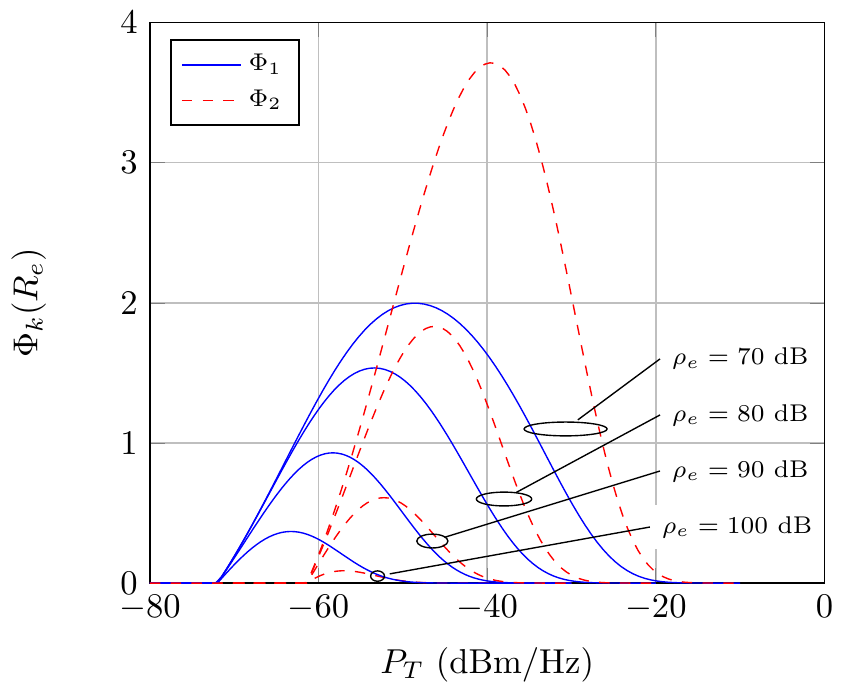}
\caption{EST of the $kth$ user versus the transmission power $P_t$ for adaptive transmission as a function of $\rho_e$, with $N=2$, $R_e=1$, $r_p=50$ m, $\lambda_e=10^{-5}$ points/m$^2$ and $\sigma_b^2=-160$ dBm/Hz.}
\label{graficaESTadaptvsPot}
\end{figure}

Fig. \ref{graficaESTadaptvsPot} shows the EST of the first NOMA user ($k=1$) versus the transmission power $P_t$ for adaptive transmission as a function of the eavesdropper density, $\lambda_e$. We observe that the optimum transmit power value is very affected by $\lambda_e$. In fact, lower eavesdropper densities lead to higher EST, although an adjustment of the transmit power is critical to achieve such maximum. For the limit case of no eavesdroppers ($\lambda_e=0$) there is no EST degradation for high $P_T$ values, as the secrecy constraint is null. 

\begin{figure}[!ht]
\centering
\includegraphics[width=0.7\columnwidth]{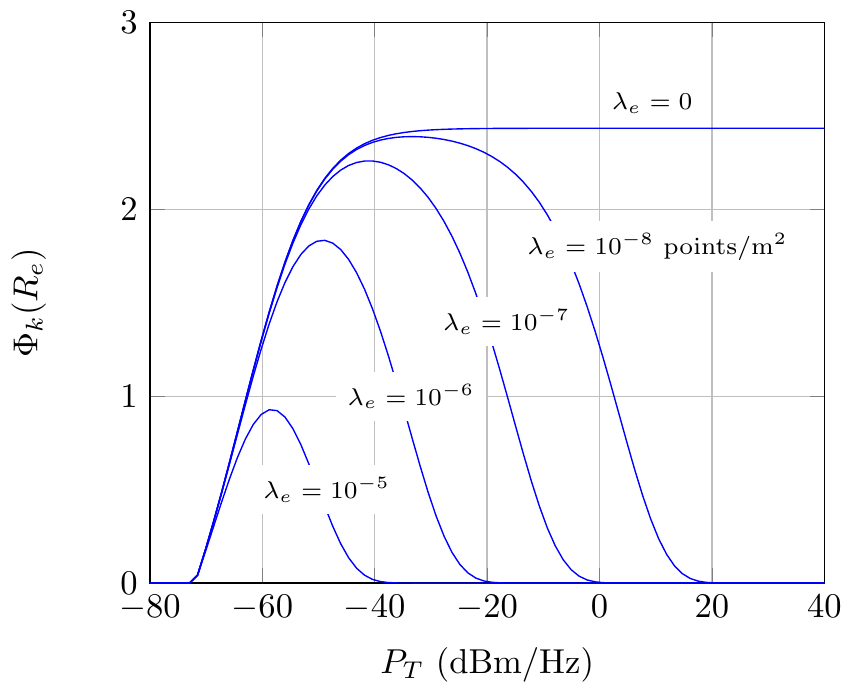}
\caption{EST of the first NOMA user ($k=1$) versus the transmission power $P_t$ for adaptive transmission as a function of $\lambda_e$, with $N=2$, $R_e=1$, $\rho_e=90$ dB, $r_p=50$ m and $\sigma_b^2=-160$ dBm/Hz.}
\label{graficaESTadaptvsPotLambda}
\end{figure}

\section{Conclusions}
In this paper, we analyzed the performance of UL NOMA for a generic number of simultaneous users, both from a connection level perspective and from a physical layer security viewpoint. We considered a passive eavesdropping scenario in which the BS and LUs are not aware of their CSI, and different cases depending on whether the LUs use a fixed or an adaptive transmission scheme. { We also considered the use of multiple antennas at the BS.} Our analysis includes the impact of an imperfect SIC during NOMA detection and an eavesdropper-exclusion radius to enhance the secrecy metrics.

We obtained new analytical expressions for the coverage probability in the uplink for LUs and eavesdroppers. In addition, we provide simple analytical expressions for the EST, which captures explicitly the reliability constraint and secrecy
constraint of wiretap channels. Our analysis allows determining the wiretap code rates that achieve the maximum EST. Performance results also help designing optimum values of the transmit power $(P_T)$ and the eavesdropper-exclusion radius $(r_p)$ in order to enhance the overall EST.


%
\appendices
\section{Proof of Lemma 1}
\label{Appendix_pk}

{  The ccdf of the SINR for the $kth$ user, $p_k(t)$, assuming $M$ antennas at the BS, can be expressed as
\begin{align}
\label{pk_MRC}
\bar{F}_{\gamma_k}(t)& =\mathbb{P}\left[ \gamma_k > t\right]   \nonumber 
\\ &
  	\overset{(a)}{=} \int_0^{r_c} \mathbb{P}
      \left[\gamma_k > t \rvert r_k\right] 
         f_{R_k}(r_k) \mathrm{d}r_k \nonumber \\
   & 	\overset{(b)}{=} \int_0^{r_c}
      \mathbb{P}\left[h_k > { t(I + \rho_b^{-1})r_k^\alpha} \rvert
       r_k\right] 
      f_{R_k}(r_k) \mathrm{d}r_k \nonumber \\
&=
 \int_0^{r_c} 
      \mathbb{E}_{I} 
      \left[
      \mathbb{P}\left[h_k > {t(i + \rho_b^{-1})r_k^\alpha}
         \rvert r_k,i\right]
      \right]
      f_{R_k}(r_k) \mathrm{d}r_k \nonumber \\&
	\overset{(c)}{=}  
\int_0^{r_c } {{\rm{e}}^{{{ - tr_k^\alpha } \mathord{\left/
 {\vphantom {{ - tr_k^\alpha } {\rho _b }}} \right.
 \kern-\nulldelimiterspace} {\rho _b }}} \mathbb{E}_{I|r_k } \left[ {{\rm{e}}^{ - tIr_k^\alpha } \left. {\sum\limits_{r = 0}^{M - 1} {\frac{{\left( { t\left( {I + \rho _b^{ - 1} } \right)r_k^\alpha } \right)^r }}{{r!}}} } \right|r_k } \right]f_{R_k } (r_k ){\rm{d}}r_k } 
\end{align}
where $(a)$ and $(b)$ follow from the total probability theorem \cite{Papoulis}, while $(c)$ follows from the fact that $H_k$ has a Gamma distribution with ccdf given by \eqref{gammaccdf}.

Using the binomial expansion $\left( {a + b} \right)^r  = \sum\limits_{k = 0}^r {{r}\choose{k}} a^{r - k} b^k$ and considering $\psi=t r_k^\alpha$, it yields
\begin{align}
\bar{F}_{\gamma_k}(t) &=
\int_0^{r_c } {{\rm{e}}^{{-\psi  \mathord{\left/
 {\vphantom {-\psi  {\rho _b }}} \right.
 \kern-\nulldelimiterspace} {\rho _b }}} } \mathbb{E}_{I|r_k } \left[ {{\rm{e}}^{-\psi I} \left. {\sum\limits_{r = 0}^{M - 1} {\sum\limits_{k = 0}^r {\frac{{\psi^r  I^k}}{{(r-k)!k!}}\left( {\frac{1}{{\rho _b }}} \right)^{r-k} } } } \right|r_k } \right]f_{R_k } (r_k ){\rm{d}}r_k 
 \nonumber \\&
=\int_0^{r_c } {{\rm{e}}^{{-\psi \mathord{\left/
 {\vphantom {-\psi  {\rho _b }}} \right.
 \kern-\nulldelimiterspace} {\rho _b }}} } \sum\limits_{r = 0}^{M - 1} {\sum\limits_{k = 0}^r {\frac{{\psi ^r }}{{(r - k)!k!}}\left( {\frac{1}{{\rho _b }}} \right)^{r - k} } } \left( {\int_0^\infty  {{\rm{e}}^{-\psi I} I^k f_I (I){\rm{d}}I} } \right)f_{R_k } (r_k ){\rm{d}}r_k 
 \nonumber \\&
= \int_0^{{r_c}} {{{\rm{e}}^{ - {{\psi} \mathord{\left/
 {\vphantom {{\psi} {{\rho _b}}}} \right.
 \kern-\nulldelimiterspace} {{\rho _b}}}}}} \sum\limits_{r = 0}^{M - 1} {\sum\limits_{k = 0}^r {\frac{{{{ \psi }^r}{{\left( { - 1} \right)}^{k}}}}{{(r - k)!k!\rho _b^{r - k}}}} } \frac{{{{\rm{d}}^k}}}{{{\rm{d}}{s^k}}}{{\cal L}_{I|{r_k}}}\left( s \right)|_{s=\psi} {f_{{R_k}}}({r_k}){\rm{d}}{r_k}
\end{align}
}

The term $\mathcal{L}_{I|r_k}(s) = \mathbb{E}_{I|r_k} 
      \left[
      \mathrm{e}^{I} 
         \rvert r_k\right]$
represents the Laplace transform of the intra-cluster interference conditioned on $r_k$, which can be expressed as
\begin{align}
\label{Li}
\mathcal{L}_{I|{r_k}}&(s) = {\mathbb{E}_{{r_j}|{r_k},{h_j}}}\left[ {\exp \left( { - s\sum\limits_{j = k + 1}^N {{h_j}r_j^{ - \alpha }} } \right)} \right] \nonumber \\
 &= {\mathbb{E}_{{r_j}|{r_k},{h_j}}}\left[ {\prod\limits_{j = k + 1}^N {{\rm{exp}}\left( { - s{h_j}r_j^{ - \alpha }} \right)} } \right]\nonumber \\
 &
 \overset{(a)}{=} \prod\limits_{j = k + 1}^N {{\mathbb{E}_{{r_j}|{r_k},{h_j}}}\left[ {{\rm{exp}}\left( { - s{h_j}r_j^{ - \alpha }} \right)} \right]} \nonumber \\
 &
 = {{\left( {{\mathbb{E}_{{r_j}|{r_k}}}\left[ \frac{1}{{1 + sr_j^{-\alpha} }} \right]} \right)}^{N - k}} 
 \nonumber \\
 &
  \overset{(b)}{=} {\left( {\int_{{r_k}}^{{r_c}} \frac{1}{{1 + sr_j^{-\alpha} }}{\frac{{2{r_j}}}{{r_c^2 - r_k^2}} } {\rm{d}}{r_j}}\right)^{N - k}}\nonumber \\
 &
 ={\left( {\frac{{2\left( {r_c^{\alpha + 2}\Omega \left( { - {{r_c^\alpha } \mathord{\left/
 {\vphantom {{r_c^\alpha } s}} \right.
 \kern-\nulldelimiterspace} s}} \right) - r_k^{\alpha + 2}\Omega \left( { - {{r_k^\alpha } \mathord{\left/
 {\vphantom {{r_k^\alpha } s}} \right.
 \kern-\nulldelimiterspace} s}} \right)} \right)}}{{s\left( {r_c^2 - r_k^2} \right)\left( {\alpha  + 2} \right)}}} \right)^{N - k}}
\end{align}
\noindent
being $\Omega \left( {{x}} \right)={{}_2{F_1}\left[ {1,\frac{{\alpha + 2}}{\alpha },2 + \frac{2}{\alpha }, x} \right]}$. Step $(a)$ comes from the fact that the fading is independent of the BPP and, although $jth$ users' location are correlated with $kth$ user when their distances are ordered, the computation of the interference can be obtained considering that the $N-k$ NOMA interfering users are located within a disk whose inner radius is $r_k$ and outer radius $r_c$. Step $(b)$ comes from the fact that the pdf of the distance from a randomly located point within that disk is given by ${f_{{R_j}\left| {{R_k}} \right.}}({r_j}\left| {{r_k}} \right.) = {{2{r_j}} \mathord{\left/
 {\vphantom {{2{r_j}} {\left( {r_c^2 - r_k^2} \right)}}} \right.
 \kern-\nulldelimiterspace} {\left( {r_c^2 - r_k^2} \right)}}$. { Note that the MRC combination does not change the distribution of the interference in our scenario, as stated in \cite{Romero08,Shah2000}.}

In \cite{Srinivasa2010}, the marginal pdf of the $kth$ nearest point to the origin of a BPP is given. In particular, this work shows that, in a BPP consisting of $N$ points randomly distributed in a 2-dimensional ball of radius $r_c$ centered at the origin, the Euclidean distance $R_k$ from the origin to its $kth$ nearest point follows a generalized beta distribution
\begin{equation}
{f_{{R_k}}}(r_k) = \frac{2}{{{r_c}}}\frac{{\Gamma \left( {k + \frac{1}{2}} \right)\Gamma \left( {N + 1} \right)}}{{\Gamma \left( k \right)\Gamma \left( {N + \frac{3}{2}} \right)}}\beta \left( {\frac{{{r_k^2}}}{{r_c^2}};k + \frac{1}{2},N - k + 1} \right)
\label{fRk}
\end{equation}

Substituting (\ref{Li}) and (\ref{fRk}) into (\ref{pk_MRC}) the proof is complete.

\section*{Acknowledgements}
This work has been supported by the Spanish Government (Ministerio de Econom\'ia y Competitividad) under grant  TEC2016-80090-C2-1-R, and Universidad de M\'alaga.

\bibliographystyle{ieeetr}
\bibliography{bibfile}

\end{document}